\documentclass[12pt]{article}
\usepackage{graphicx}
\usepackage{amsmath,amssymb}
\usepackage{amsmath}
\usepackage[latin1]{inputenc}

\newcommand{\newc}{\newcommand}
\newc{\beq}{\begin{equation}}
\newc{\eeq}{\end{equation}}
\newc{\bea}{\begin{array}}
\newc{\eea}{\end{array}}
\newcommand{\ben}{\begin{eqnarray}}
\newcommand{\een}{\end{eqnarray}}
\newc{\ra}{\rightarrow}
\newc{\bfx}{{\bf x}}
\newc{\bfV}{{\bf V}}
\newc{\cO}{{\cal O}}
\newc{\bfv}{{\bf v}}
\newc{\bfu}{{\bf u}}
\newc{\bfp}{{\bf p}}
\newc{\ve}{{\varepsilon}}
\newc{\Psibar}{\overline\Psi}
\newc{\w}{{\bf w}}
\newc{\E}{{\mathbf{E}}}
\newc{\EE}{{\mathcal E}}
\newc{\bfn}{{\mathbf\nabla}}
\newc{\la}{{\cal L}}
\newc{\tla}{{\tilde{\cal L}}}
\newc{\bp}{{\bf p}}
\newc{\ho}{\hookrightarrow }
\newc{\bP}{{\bf P}}
\newc{\pd}{{\partial}}
\newc{\piv}{{\partial_4}}
\newc{\pv}{{\partial_5}}
\newc{\bJ}{{\bf J}}
\newc{\bze}{{\mathbf 0}}
\newc{\bK}{{\bf K}}
\newc{\tphi}{{\tilde\phi}}
\newc{\tF}{{\tilde F}}
\newc{\tD}{{\tilde D}}
\newc{\tJ}{{\tilde J}}
\newc{\tj}{{\tilde j}}
\newc{\bD}{{\bf D}}
\newc{\tvphi}{{\tilde\varphi}}
\newc{\trho}{{\tilde\rho}}
\newc{\ttheta}{{\tilde\theta}}
\newc{\tpsi}{{\tilde\psi}}
\newc{\tu}{{\tilde u}}
\newc{\cD}{{\cal D}}
\newc{\tPhi}{{\tilde\Phi}}
\newc{\tPsi}{{\tilde\Psi}}
\newc{\tA}{{\tilde A}}
\newc{\talpha}{{\tilde\alpha}}
\newc{\tbeta}{{\tilde\beta}}
\newc{\bA}{{\mathbf A}}
\newc{\bB}{{\bf B}}
\newc{\br}{{\bf r}}
\newc{\sig}{{\mathbf\sigma}}
\newc{\eg}{{\rm e.g.\ }}
\newc{\ie}{{\rm i.e.\ }}
\newcommand{\bey}{\begin{eqnarray}}
\newcommand{\pslash}{\not{\hbox{\kern-2.3pt $p$}}}
\newcommand{\pdslash}{\not{\hbox{\kern-2pt $\partial$}}}
\newcommand{\eey}{\end{eqnarray}}
\newtheorem{theorem}{Theorem}
\newtheorem{proposition}{Proposition}
\newenvironment{proof}[1][Proof]{\noindent\textbf{#1.} }{\ \rule{0.5em}{0.5em}}
\newtheorem{lemma}{Lemma}
\newtheorem{definition}{Definition}

\begin{document}

\begin{titlepage}
\vskip 2cm
\begin{center}
{\Large  Information-theoretic structure for the Tsallis q-entropy in statistical physics
\footnote{{\tt matrindade@uneb.br}}}
 \vskip 10pt
{ Marco A. S. Trindade, \\}
\vskip 5pt
{\sl Colegiado de Física, Departamento de Ciências Exatas e da Terra, Universidade do Estado da Bahia\\
Rua Silveira Martins, Cabula, 41150000, Salvador, Bahia, Brazil\\}
\vskip 2pt
\end{center}

\begin{abstract}

In this work, we derive information-theoretic properties for a modified Tsallis entropy, hereinafter referred to as 
$q$-entropy. We introduce the notions of joint $q$-entropy, conditional $q$-entropy, relative $q$-entropy, conditional mutual $q$-information, and establish several inequalities analogous to those of classical information theory. Within the context of Markov chains, these results are employed to prove a version of the second law of thermodynamics. Furthermore, we investigate the maximum entropy method in this setting. Finally, we prove a Tsallis version of the Shannon-McMillan-Breiman theorem and discuss the implications of these results in nonextensive statistical physics.
\end{abstract}

\bigskip

{\it Keywords:} q-entropy; Tsallis entropy, information theory.

{\it PACS:} 05.20.-y; 89.70.Cf; 89.75.-k

\vskip 3pt

\end{titlepage}


\newpage

\setcounter{footnote}{0} \setcounter{page}{1} \setcounter{section}{0} %
\setcounter{subsection}{0} \setcounter{subsubsection}{0}

\section{Introduction}
Shannon entropy is a crucial concept in Information Theory \cite{Shannon}. It was introduced by Claude Shannon in 1948 and is a measure of average uncertainty in the random variable \cite{Cover}. As highlighted by Nielsen \cite{Nielsen} there is two complementary views: it how much information we acquired when know the value of random variable, on average, or alternatively it is a measure of uncertainty about the variable before we know its value. This definition is analogous to the statistical entropy in statistical mechanics, introduced by Boltzmann, in 1870 \cite{Huang, Landau, Pathria}.

A generalization of the Boltzmann-Gibbs entropy within the scenario of statistical mechanics was proposed by Tsallis in 1988 \cite{Tsallis}, motivated by the scaling properties observed in multifractal systems. Unlike the standard Boltzmann-Gibbs entropy, the Tsallis entropy is nonadditive, thereby violating the additivity property that constitutes one of the fundamental assumptions of Callen's third postulate of equilibrium thermodynamics. This generalized entropy is characterized by an entropic index $q$, which quantifies the degree of nonextensivity of the system. The classical Boltzmann-Gibbs entropy is recovered in the limiting case $q \rightarrow 1$, ensuring consistency with standard statistical mechanics.

Over the past decades, Tsallis statistics has been successfully applied to a broad class of physical systems in which long-range interactions, long-term memory effects, or multifractal phase-space structures play a relevant role \cite{Tsallis, T1}. Such systems frequently arise in statistical physics, plasma physics, turbulence, gravitational systems, and other complex systems \cite{Dani1, Dani2, Marco1, Marco2}, where the assumptions underlying Boltzmann-Gibbs statistics are not fully satisfied \cite{TG}. In this scenario, the nonadditive nature of Tsallis entropy provides a  theoretical scheme for describing anomalous scaling behaviors and generalized thermodynamic relations, which are not adequately captured by conventional entropic measures \cite{TG}.

On the other hand, information theory provides a set of statistical tools for the study of complex systems \cite{Lind}. The maximum entropy method, for instance, has been regarded as a foundation for the construction of a theory of complex systems \cite{Golan}. This leads to the question of whether Shannon's information theory admits a generalization based on Tsallis entropy and whether such a generalization can shed light on attempts to overcome limitations of information theory in the study of complex systems \cite{Varley}. From this perspective, theoretical information-theoretic properties of Tsallis entropies were derived by Furuichi, with the aim of constructing a non-additive information theory \cite{Furuichi}.

This work develops information-theoretic definitions and inequalities within the context of a modified Tsallis entropy, which can be regarded as a natural generalization of Shannon entropy. We introduce notions such as joint $q$-entropy, conditional $q$-entropy, relative $q$-entropy, and conditional mutual $q$-information, formulated in a manner distinct from the approach of Ref.~\cite{Furuichi}, and establish several inequalities analogous to those of classical information theory. In the context of Markov chains, the derived inequalities are employed to prove a version of the second law of thermodynamics. Furthermore, the applicability of these results to the maximum entropy method is investigated. Finally, we prove a $q$-version of the Shannon-McMillan-Breiman theorem.

The paper is organized as follows. Section 2 contains the basic definitions and inequalities. In Section 3 we perform a stochastic analysis and we derive a second law of thermodynamics in the context of Markov chain. In Section 4 we apply the maximum entropy method in this scenario. Section 6 is devoted to presenting a Tsallis version of the Shannon-McMillan-Breiman theorem. Section 7 is dedicated to conclusions and perspectives. Basic probabilistic concepts relevant to this study are present in Appendix A.

\section{Basic Results}
Let $X$ be a discrete random variable with alphabet $\chi$ and probability mass
function $\{p(x)\}_{x\in\chi}$. The Tsallis entropy \cite{Tsallis, T1, TG, Furuichi}
\begin{equation} \label{T1}
S_{q}(X) = - \sum_{x \in \chi} p^{q}(x) \ln_{q} p(x),
\end{equation}
where
\begin{equation}
\ln_{q}(x)=\frac{x^{1-q}-1}{1-q}  \  \  \ (for \ \ x>0, q \in \mathbb{R}),
\end{equation}
is not nonadditive:
\begin{equation}
S_{q}^{(A+B)}=S_{q}^{(A)}+S_{q}^{(B)}+\frac{(1-q)}{k}S_{q}^{(A)}+S_{q}^{(B)}
\end{equation}
In this work we consider the Tsallis $q$-entropy given by
\begin{equation} \label{T2}
 H_{q}(X) = - \sum_{x \in \chi} p(x) \ln_{q} p(x)
\end{equation}

Furuichi \cite{Furuichi} explores Tsallis entropy given by (\ref{T1}) in the context of information theory to construct a nonadditive information theory. Here, we derive properties of Tsallis $q$-entropy, given by (\ref{T2}), analogous to information theory and we obtain a $q$-version of second law of thermodynamics and a $q$- version of Shannon-McMillan-Breiman theorem. Therefore, we obtain several results about $q$-entropy to ensure the consistency of the formulation.

Similarly to Shannon entropy, we have the following result
\begin{proposition}
$H_{q} \geq 0$, for $q\neq 1$.
\end{proposition}

\begin{proof}
By definition, with $t\in(0,1]$, we have
\begin{eqnarray}
\ln_q t=\frac{t^{1-q}-1}{1-q}.
\end{eqnarray}
If $q<1$, then $1-q>0$ and, since $t\le 1$, one has $t^{1-q}\le 1$,
which implies $\ln_q t\le 0$.
If $q>1$, then $1-q<0$ and, again using $t\le 1$, one has
$t^{1-q}\ge 1$, which also yields $\ln_q t\le 0$.
Hence, for all $q\neq 1$ and all $t\in(0,1]$,
\begin{eqnarray}
\ln_q t\le 0,
\end{eqnarray}
with equality if and only if $t=1$. Since $p(x)\ge 0$ for all $x\in\chi$, it follows that
\begin{eqnarray}
-\,p(x)\ln_q p(x)\ge 0,
\end{eqnarray}
Therefore, every term in the sum defining $H_q(X)$ is non-negative,
and consequently $H_q(X)\ge 0$.
\end{proof}

\begin{definition}
The conditional $q$-entropy $H_{q}(Y|X)$ with $(X,Y)\sim p(x,y)$ is defined by
\begin{equation}
H_{q}(X,Y)= - \sum_{x \in \chi} \sum_{y \in \emph{y}}p(x,y) \ln_{q} p(y|x)
\end{equation}
\end{definition}
\begin{proposition}
If $(X,Y)\sim p(x,y)$ and $0\leq q<1$, we have the following inequality
\begin{equation}
H_{q}(X,Y) \geq H_{q}(X)+H_{q}(Y|X)
\end{equation}
\end{proposition}
\begin{proof}
\begin{eqnarray}
H_{q}(X,Y)&=& - \sum_{x \in \chi} \sum_{y \in \emph{y}}p(x,y) \ln_{q} p(x)p(y|x) \nonumber \\
          &=& - \sum_{x \in \chi} \sum_{y \in \emph{y}}p(x,y) \ln_{q} p(x)- \sum_{x \in \chi} \sum_{y \in \emph{y}}p(x,y) \ln_{q} p(y|x) \nonumber \\
          &-&(1-q)\sum_{x \in \chi} \sum_{y \in \emph{y}}p(x,y) \ln_{q} p(x) \ln_{q} p(y|x) \nonumber \\
          &=&H_{q}(X)+H_{q}(X \vert Y)+(1-q)\sum_{x \in \chi} \sum_{y \in \emph{y}}p(x,y) \ln_{q} p(x) \ln_{q} p(y|x) \nonumber
\end{eqnarray}
where we use property
\begin{equation} \label{somalnq}
\ln_{q}(xy)=\ln_{q}(x)+\ln_{q}(y)+(1-q)\ln_{q}(x)ln_{q}(y)
\end{equation}
Since last term is non-negative,
\begin{equation}
H_{q}(X,Y) \geq H_{q}(X)+H_{q}(Y|X)
\end{equation}
\end{proof}

\begin{proposition}
If $X$ and $Y$ are independent and $0 \leq q<1$, then
\begin{eqnarray}
H_{q}(X,Y) \geq H(X)+H(Y)
\end{eqnarray}
\end{proposition}
\begin{proof}
We have
\begin{eqnarray}
H_{q}(X,Y)&=&H_{q}(X)+H_{q}(Y) \nonumber \\
          &=&H_{q}(X)+H_{q}(Y)+(1-q)\sum_{x \in \chi} \sum_{y \in \emph{y}}p(x,y) \ln_{q} p(x) \ln_{q} p(y), \nonumber \\
          &\geq& H_{q}(X)+H_{q}(Y)
\end{eqnarray}
since $0 \leq q<1$
\end{proof}

\begin{proposition}
For $0 \leq q <1$, we have
\begin{equation}
H_{q}(X,Y|Z) \geq H_{q}(X|Z)+H_{q}(Y|X,Z)
\end{equation}
\end{proposition}
\begin{proof}
\begin{eqnarray}
H_{q}(X,Y|Z)&=&-\sum_{x \in \chi} \sum_{y \in \emph{y}} \sum_{z \in \emph{Z}} p(x,y,z) \ln_{q}p(x,y|z) \nonumber \\
            &=&-\sum_{x \in \chi} \sum_{y \in \emph{y}} \sum_{z \in \emph{Z}} p(x,y,z) \ln_{q}p(y|x,z)p(x|z) \nonumber \\
            &=&-\sum_{x \in \chi} \sum_{y \in \emph{y}} \sum_{z \in \emph{Z}} p(x,y,z) \ln_{q}p(x|z) - \sum_{x \in \chi} \sum_{y \in \emph{y}} \sum_{z \in \emph{Z}} p(x,y,z) \ln_{q}p(y|x,z) \nonumber \\
            &-&(1-q)\sum_{x \in \chi} \sum_{y \in \emph{y}} \sum_{z \in \emph{Z}} p(x,y,z) \ln_{q}p(y|x,z) \ln_{q}p(x|z) \nonumber \\
            &\geq& H_{q}(X|Z)+H_{q}(Y|X,Z),
\end{eqnarray}
since
\begin{equation}
(1-q)\sum_{x \in \chi} \sum_{y \in \emph{y}} \sum_{z \in \emph{Z}} p(x,y,z) \ln_{q}p(y|x,z) \ln_{q}p(x|z) \leq 0
\end{equation}
and using the equation (\ref{somalnq})
\end{proof}

The following Lemma will be important in our stochastic analysis, for the entropy rate
\begin{lemma}\label{lemb}
Let $X_{1}, X_{2},..., X_{n}$ random variables with joint distribution $p(x_{1},x_{2},...,x_{n})$ and $0 \leq q<1$. Then
\begin{equation}
H_{q}(X_{1}, X_{2},..., X_{n}) \geq \sum_{i=1}^{n} H_{q}(X_{i}|X_{i-1},...,X_{1})
\end{equation}
\end{lemma}
\begin{proof}
We have
\begin{eqnarray}
H_{q}(X_{1}, X_{2},..., X_{n})&=&-\sum_{x_{1},x_{2},...,x_{n}}p(x_{1},x_{2},...,x_{n}) \ln_{q} [p(x_{1},x_{2},...,x_{n})] \nonumber \\
&=&-\sum_{x_{1},x_{2},...,x_{n}}p(x_{1},x_{2},...,x_{n}) \ln_{q} \left[\prod_{i=1}^{n}p(x_{i}|x_{i-1},...,x_{1})\right] \nonumber \\
&\geq&  - \sum_{x_{1},x_{2},...,x_{n}} \sum_{i=1}^{n} p(x_{1},x_{2},...,x_{n}) \ln_{q} [p(x_{i}|x_{i-1},...,x_{1})] \nonumber \\
&=& -\sum_{i=1}^{n} \sum_{x_{1},x_{2},...,x_{n}} p(x_{1},x_{2},...,x_{n}) \ln_{q} [p(x_{i}|x_{i-1},...,x_{1})] \nonumber \\
&=& -\sum_{i=1}^{n} H_{q}(X_{i}|X_{i-1},...,X_{1})
\end{eqnarray}
\end{proof}

The following definitions are similar to those concerning the Shannon entropy and they are fundamental in our formulation.
\begin{definition}
The relative $q$-entropy between two probability functions $p(x)$ and $r(x)$ is defined by
\begin{equation}
D_{q}(p||r)=\sum_{x \in \chi}p(x) \ln_{q} \left[\frac{p(x)}{r(x)}\right]
\end{equation}
\end{definition}

\begin{definition}
The mutual $q$-information $I_{q}(X;Y)$ is defined as the relative entropy between the joint distribution $p(x,y)$ and the product distribution $p(x)p(y)$:
\begin{equation}
I_{q}(X;Y)=\sum_{x \in \chi} \sum_{y \in \emph{y}}p(x,y)  \ln_{q} \left[\frac{p(x,y)}{p(x)p(y)}\right]
\end{equation}
where $p(x)$ and $p(y)$ are the marginal probability functions.
\end{definition}

\begin{definition}
The conditional mutual $q$-information of random variables $X$, $Y$ and $Z$ is defined as
\begin{equation}
I_{q}(X;Y|Z)=\sum_{x,y,z}p(x,y,z) \ln_{q} \left[\frac{p(x,y|z)}{p(x|z)p(y|z)}\right]
\end{equation}
\end{definition}

We have the chain rule for $q$-information

\begin{proposition}
\begin{eqnarray}
I_{q}(X_{1},X_{2},...,X_{n},Y)&=&\sum_{i=1}^{n}I_{q}(X_{i};Y|X_{i-1}, X_{i-2},...X_{1}) \nonumber \\
                              &=&(1-q)\sum_{x_{1},x_{2},...,x_{n},y}p(x_{1},x_{2},...,x_{n},y) \nonumber \\
&\times& \sum_{i=1}^{n} ( \ln_{q} \left[\frac{p(x_{i};y|x_{i-1}, x_{i-2},...x_{1})}{p(x_{i}|x_{i-1}, x_{i-2},...x_{1})p(y|x_{i-1}, x_{i-2},...x_{1})}\right] \nonumber \\
& \times & \ln_{q} \left [\prod_{i'=i+1}^{n}\left[\frac{p(x_{i'};y|x_{i'-1}, x_{i'-2},...x_{1})}{p(x_{i'}|x_{i'-1}, x_{i'-2},...x_{1})p(y|x_{i'-1}, x_{i'-2},...x_{1})}\right]\right])   \nonumber
\end{eqnarray}
\end{proposition}
\begin{proof}
We have that
\begin{eqnarray}
I_{q}(X_{1},X_{2},...,X_{n},Y)&=& \sum_{x_{1},x_{2},...,x_{n},y}p(x_{1},x_{2},...,x_{n},y)\ln_{q}\left[\frac{p(x_{1},x_{2},...,x_{n},y)}{p(x_{1},x_{2},...,x_{n})p(y)}\right] \nonumber \\
&=&\sum_{x_{1},x_{2},...,x_{n},y}p(x_{1},x_{2},...,x_{n},y) \nonumber \\
&\times& \ln_{q}\prod_{i=1}^{n}\left[\frac{p(x_{i};y|x_{i-1}, x_{i-2},...x_{1})}{p(x_{i}|x_{i-1}, x_{i-2},...x_{1})p(y|x_{i-1}, x_{i-2},...x_{1})}\right] \nonumber \\
&=& \sum_{x_{1},x_{2},...,x_{n},y} \sum_{i=1}^{n}p(x_{1},x_{2},...,x_{n},y) \nonumber \\
&\times& \ln_{q}\left[\frac{p(x_{i};y|x_{i-1}, x_{i-2},...x_{1})}{p(x_{i}|x_{i-1}, x_{i-2},...x_{1})p(y|x_{i-1}, x_{i-2},...x_{1})}\right]
\end{eqnarray}
\begin{eqnarray}
&+&(1-q)\sum_{x_{1},x_{2},...,x_{n},y}p(x_{1},x_{2},...,x_{n},y) \nonumber \\
&\times& \sum_{i=1}^{n} \ln_{q} \left[\frac{p(x_{i};y|x_{i-1}, x_{i-2},...x_{1})}{p(x_{i}|x_{i-1}, x_{i-2},...x_{1})p(y|x_{i-1}, x_{i-2},...x_{1})}\right] \nonumber \\
& \times & \ln_{q} \left [\prod_{i'=i+1}^{n}\left[\frac{p(x_{i'};y|x_{i'-1}, x_{i'-2},...x_{1})}{p(x_{i'}|x_{i'-1}, x_{i'-2},...x_{1})p(y|x_{i'-1}, x_{i'-2},...x_{1})}\right]\right] \nonumber \\
&=&\sum_{i=1}^{n}I_{q}(X_{i};Y|X_{i-1}, X_{i-2},...X_{1}) \nonumber \\
&+&(1-q)\sum_{x_{1},x_{2},...,x_{n},y}p(x_{1},x_{2},...,x_{n},y) \nonumber \\
&\times& \sum_{i=1}^{n} \ln_{q} \left[\frac{p(x_{i};y|x_{i-1}, x_{i-2},...x_{1})}{p(x_{i}|x_{i-1}, x_{i-2},...x_{1})p(y|x_{i-1}, x_{i-2},...x_{1})}\right] \nonumber \\
& \times & \ln_{q} \left [\prod_{i'=i+1}^{n}\left[\frac{p(x_{i'};y|x_{i'-1}, x_{i'-2},...x_{1})}{p(x_{i'}|x_{i'-1}, x_{i'-2},...x_{1})p(y|x_{i'-1}, x_{i'-2},...x_{1})}\right]\right] \nonumber
\end{eqnarray}
\end{proof}

The next Lemma is a consequence of the concavity of the $q$-logarithm, and it is the analogous the classical log sum inequality \cite{Cover}. A similar result was obtained by Furuichi \cite{Furuichi}.
\begin{lemma}
($q$-ln sum inequality) Let $r_{1},r_{2},...,r_{n}$ and $s_{1},s_{2},...,s_{n}$ nonnegative numbers. Then
\begin{equation}
\sum_{i=1}^{n} r_{i} \ln_{q}\left(\frac{r_{i}}{s_{i}}\right) \geq \left(\sum_{i=1}^{n} r_{i}\right)\ln_{q}\left(\frac{\sum_{i=1}^{n}r_{i}}{\sum_{i=1}^{n}s_{i}}\right)
\end{equation}
with equality if and only $\frac{r_{i}}{s_{i}}=c$, with $c=constant$.
\end{lemma}

\begin{proof}
Note that $f(u)=u \ln_{q}(u)$ is strictly convex for $q \leq 2$:
\begin{equation}
f"(u)=\frac{(2-q)(1-q)p^{1-q-1}}{1-q}=(2-q)p^{-q}\leq 0
\end{equation}
Therefore the Jensen's inequality provides
\begin{equation}
\sum_{i}\beta_{i}f(u_{i})\geq f(\sum_{i}\beta_{i}u_{i}),
\end{equation}
with $ \sum_{i}\beta_{i}=1$ and $\beta_{i}\geq 0$. The conclusion follows setting $\beta_{i}=\frac{s_{i}}{\sum_{j=1}^{n}s_{j}}$ and $u_{i}=\frac{r_{i}}{s_{i}}$, so that
\begin{equation}
\sum_{i=1}^{n} \frac{r_{i}}{\sum_{j=1}^{n}s_{j}} \ln_{q}\left(\frac{r_{i}}{s_{i}}\right) \geq \left(\sum_{i=1}^{n}\frac{r_{i}}{\sum_{j=1}^{n}s_{j}}\right)\ln_{q} \left(\sum_{i=1}^{n}\frac{r_{i}}{\sum_{j=1}^{n}s_{j}}\right)
\end{equation}
\end{proof}

Next, we have the $q$-information inequality.
\begin{theorem}\label{naoneg}
For $q \leq 2$, the relative $q$-entropy satisfies
\begin{equation}
D_{q}(r(x)||s(x)) \geq 0
\end{equation}
where $r(x)$ and $s(x)$ are two probability functions and we have the equality if and only if $r(x)=s(x)$, for all $x$.
\end{theorem}

\begin{proof}
Using that $q$-ln sum inequality,
\begin{eqnarray}
D_{q}(r(x)||s(x))&=&\sum_{x}r(x)\ln_{q}\left(\frac{r(x)}{s(x)}\right) \nonumber \\
                 &\geq& \left(\sum_{x}r(x)\right)\ln_{q}\left(\frac{\sum_{x}r(x)}{\sum_{x}s(x)}\right) \nonumber \\
                 &=0&,
\end{eqnarray}
where the equality occurs if and only if $r(x)=s(x)$, once $r(x)$ and $s(x)$ are probability functions.
\end{proof}

\begin{theorem}
Let $X$ a random variable with probability function $p(x)$ and $q \leq 2$. Then
\begin{equation}
H_{q} \leq \ln_{q} |\chi|,
\end{equation}
where $|\chi|$ denotes the numbers of elements in the range of $X$ and the equality is satisfied if and only $X$ has a uniform distribution.
\end{theorem}
\begin{proof}
Consider $r(x)=|\chi|^{-1}$ the uniform probability function. Thus
\begin{eqnarray}
D_{q}(p(x)||r(x))&=&\sum_{x}p(x)\ln_{q}\left(p(x)|\chi|\right) \nonumber \\
&=&\sum_{x}p(x)\ln_{q}p(x)\left[1+(1-q)\ln_{q}|\chi|\right]+\ln_{q}|\chi| \nonumber \\
&\geq& 0
\end{eqnarray}
by nonnegativity of relative $q$-entropy. Therefore
\begin{equation}
-\sum_{x}p(x)\ln_{q}p(x) \leq \ln_{q} |\chi|
\end{equation}
as desired.
\end{proof}

The next theorem is analogous to the data processing inequality in classical information theory \cite{Cover}
\begin{theorem}
If $X \rightarrow Y \rightarrow Z$ and $0 \leq q <1$ (Markov chain), then
\begin{eqnarray}
I_{q}(X;Y) \geq I_{q}(X;Z)+(1-q)\sum_{x,y,z}p(x,y,z)\ln_{q}\frac{p(x,z)}{p(x)p(z)} \ln_{q}\frac{p(x,y,z)}{p(x|z)p(y|z)} \nonumber
\end{eqnarray}
\end{theorem}
\begin{proof}
We have that
\begin{eqnarray}
I_{q}(X;Y,Z)&=&\sum_{x,y,z}p(x,y,z)\ln_{q}\frac{p(x,y,z)}{p(x)p(y,z)} \nonumber \\
            &=&\sum_{x,y,z}p(x,y,z)\ln_{q}\frac{p(x,z)p(x,y|z)}{p(x)p(z)p(x|z)p(y|z)} \nonumber \\
            &=&\sum_{x,y,z}p(x,y,z)\ln_{q}\frac{p(x,z)}{p(x)p(z)}+\sum_{x,y,z}p(x,y,z)\ln_{q}\frac{p(x,y|z)}{p(x|z)p(y|z)} \nonumber \\
            &+&(1-q)\sum_{x,y,z}p(x,y,z)\ln_{q}\frac{p(x,z)}{p(x)p(z)}\ln_{q}\frac{p(x,y|z)}{p(x|z)p(y|z)} \nonumber \\
            &=&I_{q}(X;Z)+I_{q}(X;Y|Z)+(1-q)\sum_{x,y,z}p(x,y,z)\ln_{q}\frac{p(x,z)}{p(x)p(z)}\ln_{q}\frac{p(x,y|z)}{p(x|z)p(y|z)} \nonumber
\end{eqnarray}
Similarly, one can prove that
\begin{eqnarray}
I_{q}(X;Y,Z)=I_{q}(X;Y)+I_{q}(X;Z|Y)+(1-q)\sum_{x,y,z}p(x,y,z)\ln_{q}\frac{p(x,y)}{p(x)p(y)}\ln_{q}\frac{p(x,z|y)}{p(x|y)p(z|y)} \nonumber
\end{eqnarray}
Note that $I_{q}(X;Y|Z)=0$, since $X$ and $Z$ are conditionally independent given $Y$. Hence
\begin{eqnarray}
I_{q}(X;Y)= I_{q}(X;Z)+I_{q}(X;Y|Z)+(1-q)\sum_{x,y,z}p(x,y,z)\ln_{q}\frac{p(x,z)}{p(x)p(z)}\ln_{q}\frac{p(x,y|z)}{p(x|z)p(y|z)} \nonumber
\end{eqnarray}
The proof is finished since
\begin{equation}
I_{q}(X;Y|Z) \geq 0
\end{equation}
and
\begin{equation}
\sum_{x,y,z}p(x,y,z)\ln_{q}\frac{p(x,z)}{p(x)p(z)}\ln_{q}\frac{p(x,y|z)}{p(x|z)p(y|z)} \leq 0
\end{equation}
\end{proof}

\section{Stochastic analysis and the second law of thermodynamics}

The following chain rule for relative $q$-entropy is used to prove the violation of the second law of thermodynamics analogously to classical analysis to Shannon relative entropy where the second law is derived \cite{Cover} (in this case, it is not violated).

\begin{proposition}\label{propb}
For $0 \leq q <1$, the relative $q$-entropy satisfies
\begin{eqnarray}
D_{q}(p(x,y)||r(x,y))&=&D_{q}((p(x)||r(x))+D_{q}((p(y|x)||r(y|x)) \nonumber \\
                     &+&(1-q)\sum_{x,y}p(x,y)\ln_{q}\frac{p(x)}{r(x)}\ln_{q}\frac{p(y|x)}{r(y|x)}
\end{eqnarray}
\end{proposition}
\begin{proof}
\begin{eqnarray}
D_{q}(p(x,y)||r(x,y))&=&\sum_{x,y}p(x,y)\ln_{q}\frac{p(x,y)}{r(x,y)} \nonumber \\
                     &=&\sum_{x,y}p(x,y)\ln_{q}\frac{p(x)p(y|x)}{r(x)r(y|x)} \nonumber \\
                     &=&\sum_{x,y}p(x,y)\ln_{q}\frac{p(x)}{r(x)}+\sum_{x,y}p(x,y)\ln_{q}\frac{p(y|x)}{r(y|x)} \nonumber \\
                     &+&(1-q)\sum_{x,y}p(x,y)\ln_{q}\frac{p(x)}{r(x)}\ln_{q}\frac{p(y|x)}{r(y|x)}
\end{eqnarray}
\end{proof}

Now, we show that our formulation can be extended to a stochastic context.
\begin{definition}
The $q$-entropy of stochastic process $\{X_{i}\}$ is defined as
\begin{equation}
H_{q}(\chi)=\lim_{n\rightarrow\infty} \frac{1}{n}H_{q}(X_{1},X_{2},...,X_{n}),
\end{equation}
when the limit exists.
\end{definition}

Also we define the conditional $q$-entropy of stochastic process $\{X_{i}\}$
\begin{equation}
H'_{q}(\chi)=\lim_{n\rightarrow\infty}\frac{1}{n} \sum_{i=1}^{n}H_{q}(X_{n}|X_{n-1},X_{n-2},...,X_{1})
\end{equation}
when the limit exists.

\begin{theorem}
For an arbitrary stochastic process $\{X_{i}\}$, we have
\begin{equation}
H'_{q}(\chi) \leq H_{q}(\chi)
\end{equation}
\end{theorem}
\begin{proof}
Using the Lemma \ref{lemb}, we obtain
\begin{eqnarray}
\sum_{i=1}^{n}H_{q}(X_{i}|X_{i-1},X_{i-2},...,X_{1}) &\leq& H_{q}(X_{1}, X_{2},...,X_{n}) \nonumber \\
\frac{1}{n}\sum_{i=1}^{n}H_{q}(X_{i}|X_{i-1},X_{i-2},...,X_{1}) &\leq& \frac{1}{n}H_{q}(X_{1}, X_{2},...,X_{n})
\end{eqnarray}
Thus
\begin{eqnarray}
\lim_{n\rightarrow\infty}\frac{1}{n} \sum_{i=1}^{n}H_{q}(X_{i}|X_{i-1},X_{i-2},...,X_{1}) &\leq& \lim_{n\rightarrow\infty}\frac{1}{n}H_{q}(X_{1},X_{2},...,X_{n}) \nonumber \\
&=&H_{q}(\chi)
\end{eqnarray}
\end{proof}

Particularly, for a stationary Markov chain
\begin{equation}
H'_{q}(\chi)=H_{q}(X_{i+1}|X_{i}) \leq H_{q}(X)
\end{equation}
Although Maxwell's demon was proposed more than a century ago, it remains a conceptually relevant problem in the foundations of physics \cite{Max,Sz,Fey, Land, Ben, Ben1}. Aquino \cite{Aquino} analyzes the implications of adopting a non-extensive thermodynamics, showing that, in this case, the effect of Maxwell's demon would be determined by the memory of the system and would therefore be temporary, in contrast with dynamical approaches based on Lévy statistics. We believe that the parameter $q$ of non-extensive entropy may provide a theoretical perspective for investigating apparent deviations from the standard formulation of the second law of thermodynamics in generalized scenarios. The following theorem illustrates this possibility through an analysis analogous to the proof of the second law of thermodynamics presented by Cover \cite{Cover} for an isolated system modeled by a Markov chain.
\begin{theorem}
Let $\psi_{n}$ and $\psi_{n}^{'}$ be two probability distributions on the state space of a Markov chain (probability function $p(x_{n},x_{n+1})=p(x_{n})r(x_{n+1},x_{n})$ where the $r$ is the probability transition function for Markov chain), and $0 \leq q <1$. Then
\begin{eqnarray}
H_{q}(\psi_{n+1})-H_{q}(\psi_{n}) \geq \frac{1-q}{[1+(1-q)\ln_{q}|\chi|]}\sum_{n}p(x_{n+1},x_{n})\ln_{q}[p(x_{n+1})|\chi|]\ln_{q}[p(x_{n},x_{n+1})|\chi|] \nonumber \end{eqnarray}
where $\chi$ is the uniform distribution
\end{theorem}
\begin{proof}
 We consider $\psi_{n}$ and $\psi_{n}^{'}$ be two probability distributions on the state space of a Markov chain with probability functions $p(x_{n},x_{n+1})=p(x_{n})r(x_{n+1},x_{n})$ and $s(x_{n},x_{n+1})=s(x_{n})r(x_{n+1},x_{n})$. By the proposition \ref{propb}, the relative $q$-entropy can be written as
 \begin{eqnarray}
 D_{q}(p(x_{n},x_{n+1})||s(x_{n},x_{n+1}))&=&D_{q}(p(x_{n})||s(x_{n}))+D_{q}(p(x_{n+1}||x_{n})||s(x_{n+1}||x_{n})) \nonumber \\
 &+&(1-q) \sum_{n}p(x_{n},x_{n+1})\ln_{q}\left[\frac{p(x_{n})}{s(x_{n})}\right]\ln_{q}\left[\frac{p(x_{n+1}|x_{n})}{s(x_{n+1}|x_{n})}\right] \nonumber \\
 &=&D_{q}(p(x_{n+1})||s(x_{n+1}))+D_{q}(p(x_{n}||x_{n+1})||s(x_{n}||x_{n+1})) \nonumber \\
 &+&(1-q)
 \sum_{n}p(x_{n+1},x_{n})\ln_{q}\left[\frac{p(x_{n+1})}{s(x_{n+1})}\right]\ln_{q}\left[\frac{p(x_{n}|x_{n+1})}{s(x_{n}|x_{n+1})}\right] \nonumber
 \end{eqnarray}
 But
 \begin{eqnarray}
 p(x_{n+1}|x_{n})=s(x_{n+1}|x_{n})=r(x_{n+1}|x_{n}).
 \end{eqnarray}
 Thus, by Theorem \ref{naoneg} (nonnegativity of relative q-entropy)
\begin{eqnarray}
 D_{q}(p(x_{n})||s(x_{n})) &\geq& D_{q}(p(x_{n+1})||s(x_{n+1})) \nonumber \\
 &+&(1-q)\sum_{n}p(x_{n+1},x_{n})\ln_{q}\left[\frac{p(x_{n+1})}{s(x_{n+1})}\right]\ln_{q}\left[\frac{p(x_{n}|x_{n+1})}{s(x_{n}|x_{n+1})}\right] \nonumber
\end{eqnarray}
namely
\begin{eqnarray}
 D_{q}(\psi_{n}||\psi_{n}^{'}) &\geq& D_{q}(\psi_{n+1})||\psi_{n+1}^{'}) \nonumber \\
 &+&(1-q)\sum_{n}p(x_{n+1},x_{n})\ln_{q}\left[\frac{p(x_{n+1})}{s(x_{n+1})}\right]\ln_{q}\left[\frac{p(x_{n}|x_{n+1})}{s(x_{n}|x_{n+1})}\right] 
\end{eqnarray}
If $\psi_{n}^{'}=\psi$ is a stationary distribution
\begin{eqnarray}
 D_{q}(\psi_{n}||\psi) &\geq& D_{q}(\psi_{n+1})||\psi) \nonumber \\
 &+&(1-q)\sum_{n}p(x_{n+1},x_{n})\ln_{q}\left[\frac{p(x_{n+1})}{s(x_{n+1})}\right]\ln_{q}\left[\frac{p(x_{n}|x_{n+1})}{s(x_{n}|x_{n+1})}\right] 
\end{eqnarray}
Particularly, for an uniform stationary distribution
\begin{eqnarray}
D_{q}(\psi_{n}||\psi)&=&\ln_{q}|\chi|-[1+(1-q)]\ln_{q}|\chi|H_{q}(\psi_{n}) \nonumber \\
&\geq&D_{q}(\psi_{n+1}||\psi)+(1-q)\sum_{n}p(x_{n+1},x_{n})\ln_{q}\left[\frac{p(x_{n+1})}{s(x_{n+1})}\right]\ln_{q}\left[\frac{p(x_{n}|x_{n+1})}{s(x_{n}|x_{n+1})}\right] \nonumber \\
&=&\ln_{q}|\chi|-[1+(1-q)]\ln_{q}|\chi|H_{q}(\psi_{n+1}) \nonumber \\
&+&(1-q)\sum_{n}p(x_{n+1},x_{n})\ln_{q}\left[\frac{p(x_{n+1})}{s(x_{n+1})}\right]\ln_{q}\left[\frac{p(x_{n}|x_{n+1})}{s(x_{n}|x_{n+1})}\right] 
\end{eqnarray}
Therefore
\begin{equation}
[H_{q}(\psi_{n+1})-H_{q}(\psi_{n})][1+(1-q)\ln_{q}|\chi|] \geq T_{q},
\end{equation}
where
\begin{equation}
T_{q}=(1-q)\sum_{n}p(x_{n+1},x_{n})\ln_{q}\left[\frac{p(x_{n+1})}{s(x_{n+1})}\right]\ln_{q}\left[\frac{p(x_{n}|x_{n+1})}{s(x_{n}|x_{n+1})}\right] \end{equation}
\end{proof}

Notice that the term $T_{q}$ may be negative, allowing a decrease in entropy. 

\section{Maximum entropy method for the q-entropy}
The Maximum Entropy method was originally formulated by Edwin T. Jaynes \cite{Jay, Cover}. Its central idea is to select, among all probability distributions compatible with a given set of macroscopic constraints (such as known expectation values). In the classical scenario, the Shannon entropy is employed, and its maximization naturally leads to the familiar exponential distributions of the statistical mechanics developed by Boltzmann and Gibbs \cite{Jay}.
In this context, we can apply the maximum entropy method \cite{Jay, Cover} for $q$- entropy with the constraints:
\begin{equation}
\sum_{i=1}^{n} p_{i}\epsilon_{i}=\widehat{\epsilon}
\end{equation}
and
\begin{equation}
\sum_{i=1}^{n}p_{i}=1
\end{equation}
The lagrangian is given by
\begin{equation}
L=-\sum_{i=1}^{n}p_{i}\ln_{q}p_{i}-(\lambda -1)(\sum_{i=1}^{n}p_{i}-1)- \mu(\sum_{i=1}^{n} p_{i}\epsilon_{i}-\widehat{\epsilon})
\end{equation}
Thus we have the associated probability distribution
\begin{equation}
p_{i}=exp_{q}\left(\frac{- \lambda-\mu\epsilon_{i}}{2-q}\right)
\end{equation}
The parameters $\lambda$ and $\mu$ can be obtained through system of nonlinear equations

\begin{displaymath}
y = \left\{ \begin{array}{ll}
\sum_{i=1}^{n}exp_{q}\left(\frac{- \lambda-\mu\epsilon_{i}}{2-q}\right)\epsilon_{i}=\widehat{\epsilon} & \\
\sum_{i=1}^{n}exp_{q}\left(\frac{- \lambda-\mu\epsilon_{i}}{2-q}\right)=1 &
\end{array} \right.
\end{displaymath}

It remains for us to prove that we have the maximum entropy distribution. Let
\begin{equation}
H_{q}'=-\sum_{i=1}^{n}f_{i}\ln_{q}f_{i}
\end{equation}
be another probability arbitrary distribution. We have that,
\begin{eqnarray}
H_{q}-H_{q}'&=&-\sum_{i=1}^{n}p_{i}\ln_{q}p_{i}+\sum_{i=1}^{n}f_{i}\ln_{q}f_{i} \nonumber \\
&=&(f_{i}-p_{i})\left(\frac{-\lambda-\mu \epsilon_{i}}{2-q}\right)+\sum_{i=1}^{n}[1+(1-q)\ln_{q}p_{i}]f_{i}\ln_{q}\left(\frac{f_{i}}{p_{i}}\right) \nonumber \\
\end{eqnarray}
Then
\begin{eqnarray}
H_{q}-H_{q}'&=&-\sum_{i=1}^{n}[1+(1-q)\ln_{q}p_{i}]f_{i}\ln_{q}\left(\frac{f_{i}}{p_{i}}\right) \nonumber \\
&\geq& \sum_{i=1}^{n}f_{i}\ln_{q}\left(\frac{\sum_{i=1}^{n}f_{i}}{\sum_{i=1}^{n}p_{i}}\right) \nonumber \\
&=&0,
\end{eqnarray}
by using $q$-ln sum inequality (Lemma 2). Besides, we have the equality  if and only if $f_{i}=p_{i}$.

\section{q-Version Shannon-McMillan-Breiman theorem}
In this section, we present a version of Shannon-McMillan-Breiman theorem \cite{Cover} within the nonadditive entropic scenario considered in this work. The result characterizes the asymptotic behavior of information associated with stationary stochastic processes and establishes a $q$--version of the asymptotic equipartition property for general ergodic process. In this section, we consider $1/2<q<1$, corresponding to systems in which rare microstates are moderately amplified but do not dominate, consistent with ergodic behavior and effective additivity at large scales. Some basic probabilistic concepts relevant to this section are present in Appendix A

\begin{lemma}\label{lem3}
Let $\{X_n\}_{n\in\mathbb{Z}}$ be a stationary ergodic stochastic process over a finite alphabet $\mathcal{X}$
and $1/2q<1$.
Define the non-extensive conditional entropy rate:
\begin{eqnarray}
H_{q,\infty} := \mathbb{E}\big[-\ln_q p(X_0 \mid X_{-\infty}^{-1})\big],
\end{eqnarray}
where the expectation is over the joint distribution of the past $X_{-\infty}^{-1}$ and the present $X_0$.
Using the rule of probability conditioning,
\begin{eqnarray}
p(x_{-\infty}^{-1}, x_0) = p(x_{-\infty}^{-1}) \, p(x_0 \mid x_{-\infty}^{-1}),
\end{eqnarray}
we can expand $H_{q,\infty}$ as a double sum over the alphabet $\mathcal{X}$ and past sequences:
\begin{eqnarray}
H_{q,\infty}
= - \sum_{x_{-\infty}^{-1}} p(x_{-\infty}^{-1}) \sum_{x_0 \in \mathcal{X}} p(x_0 \mid x_{-\infty}^{-1}) 
\frac{p(x_0 \mid x_{-\infty}^{-1})^{1-q}-1}{1-q}.
\end{eqnarray}
Then, almost surely,
\begin{eqnarray}
\limsup_{n\to\infty} \Big(-\frac{1}{n} \ln_q p(X_0^{\,n-1}\vert X_{-\infty}^{-1})\Big) \le H_{q,\infty}.
\end{eqnarray}

\end{lemma}

\begin{proof}
Let
\begin{eqnarray}
Z_i := -\ln_q p(X_i \mid X_{-\infty}^{\,i-1}).
\end{eqnarray}
Since $\{X_i\}$ is stationary and ergodic, $\{Z_i\}$ is also stationary and ergodic.  
By the ergodic theorem \cite{Cover}:
\begin{eqnarray} \label{erg}
\frac{1}{n} \sum_{i=0}^{n-1} Z_i \;\xrightarrow{\text{a.s.}}\; \mathbb{E}[Z_0] = H_{q,\infty}.
\end{eqnarray}
We have 
\begin{eqnarray}
\ln_q(xy)=\ln_q(x)+\ln_q(y)+(1-q)\ln_q(x)\ln_q(y),
\end{eqnarray}
so that for any $n$, with $0 \leq q <1$,
\begin{equation}
-\ln_q p(X_0^{n-1}\vert X_{-1}^{\infty}) = -\ln_q \prod_{i=0}^{n-1}
p\bigl(X_i \mid X_{-\infty}^{\,i-1}\bigr) \le \sum_{i=0}^{n-1} Z_i
\end{equation}
Dividing by $n$:
\begin{eqnarray}
-\frac{1}{n} \ln_q p(X_0^{n-1}\vert X_{-\infty}^{-1}) \le \frac{1}{n} \sum_{i=0}^{n-1} Z_i.
\end{eqnarray}
From (\ref{erg}), taking $\limsup$ gives
\begin{eqnarray}
\limsup_{n \to \infty} \Big(-\frac{1}{n} \ln_q p(X_0^{n-1}\vert X_{-\infty}^{-1})\Big) \le H_{q,\infty}.
\end{eqnarray}
Using the expanded form of $H_{q,\infty}$:
\begin{eqnarray}
\limsup_{n \to \infty} \Big(-\frac{1}{n} \ln_q p(X_0^{n-1}\vert X_{-\infty}^{-1})\Big) \le 
- \sum_{x_{-\infty}^{-1}} p(x_{-\infty}^{-1}) \sum_{x_0 \in \mathcal{X}} p(x_0 \mid x_{-\infty}^{-1}) \frac{p(x_0 \mid x_{-\infty}^{-1})^{1-q}-1}{1-q}.
\end{eqnarray}
Here the first $p$ comes from the distribution of the past, and the second $p$ comes from the conditional probability of $X_0$ given the past.
\end{proof}

The $k$th-order Markov approximation to the probability is defined for $n \ge k$ as \cite{Cover}
\begin{eqnarray}
p_k(X_0^{n-1}) = p(X_0^{k-1}) \prod_{i=k}^{n-1} p\bigl(X_i \mid X_{i-k}^{i-1}\bigr).
\end{eqnarray}

\begin{lemma}\label{lem4}
Let $\{X_n\}_{n\in\mathbb{Z}}$ be a stationary stochastic process with finite alphabet. 
Assume in addition that there exists  constants $C>0$ and $C'>0$ such that 
\begin{eqnarray}\label{sup1}
\sup_{n \geq 1}\mathbb{E}
\left[
\frac{p(X_0^{n-1}\mid X_{-\infty}^{-1})}
     {p(X_0^{n-1})} 
\right]<C
\end{eqnarray}
and
\begin{eqnarray}\label{sup2}
\sup_{n \geq 1}\mathbb{E}
\left[
\frac{p(X_0^{n-1})}
     {p_k(X_0^{n-1})}
\right]<C'
\end{eqnarray}
Then, almost surely,
\begin{eqnarray}
\limsup_{n\to\infty}
\frac{1}{n}
\ln_q
\frac{p(X_0^{n-1}\mid X_{-\infty}^{-1})}
     {p(X_0^{n-1})}
\le 0,
\end{eqnarray}
\begin{eqnarray}
\limsup_{n\to\infty}
\frac{1}{n}
\ln_q
\frac{p(X_0^{n-1})}
     {p_k(X_0^{n-1})}
\le 0.
\end{eqnarray}
\end{lemma}
\begin{proof}

By Markov inequality \cite{Durret}, for any $t_n>0$,
\begin{eqnarray}
\Pr\left\{
\frac{p(X_0^{n-1}\mid X_{-\infty}^{-1})}
     {p(X_0^{n-1})}
\ge t_n
\right\}
&\le&
\frac{C}{t_n},
\end{eqnarray}
\begin{eqnarray}
\Pr\left\{
\frac{p(X_0^{n-1})}
     {p_k(X_0^{n-1})}
\ge t_n
\right\}
&\le&
\frac{C'}{t_n}.
\end{eqnarray}
Since $\frac{d}{dx}\ln_q(x)>0$ for $x>0$, we obtain
\begin{eqnarray}
\Pr\left\{
\frac{1}{n}
\ln_q
\frac{p(X_0^{n-1}\mid X_{-\infty}^{-1})}
     {p(X_0^{n-1})}
\ge
\frac{1}{n}\ln_q t_n
\right\}
&\le&
\frac{C}{t_n},
\end{eqnarray}
\begin{eqnarray}
\Pr\left\{
\frac{1}{n}
\ln_q
\frac{p(X_0^{n-1})}
     {p_k(X_0^{n-1})}
\ge
\frac{1}{n}\ln_q t_n
\right\}
&\le&
\frac{C'}{t_n}.
\end{eqnarray}
Choose $t_n = n^2$. Then $\sum_{n=1}^{\infty} 1/t_n < \infty$ and for $q>1/2$,
\begin{eqnarray}
\ln_q(t_n)
=
\ln_q(n^2)
=
\frac{n^{2(1-q)}-1}{1-q},
\end{eqnarray}
\begin{eqnarray}
\frac{1}{n}\ln_q(t_n) \longrightarrow 0.
\end{eqnarray}
By the Borel-Cantelli lemma \cite{Durret},
\begin{eqnarray}
\limsup_{n\to\infty}
\frac{1}{n}
\ln_q
\frac{p(X_0^{n-1}\mid X_{-\infty}^{-1})}
     {p(X_0^{n-1})}
\le 0 \quad a.s.,
\end{eqnarray}
and
\begin{eqnarray}
\limsup_{n\to\infty}
\frac{1}{n}
\ln_q
\frac{p(X_0^{n-1})}
     {p_k(X_0^{n-1})}
\le 0 \quad a.s.
\end{eqnarray}
\end{proof}

The supremum conditions (\ref{sup1}) and (\ref{sup2}) provide uniform control over block probability ratios and may be satisfied in stationary processes with memory decay, allowing the asymptotic analysis.

\begin{lemma}\label{lem6}
Let $\{X_n\}_{n\in\mathbb{Z}}$ be a stationary ergodic stochastic process over a finite alphabet. Fix $k \ge 1$ and define
\begin{eqnarray}
p_k(X_0^{n-1}) = p(X_0^{k-1}) \prod_{i=k}^{n-1} p(X_i \mid X_{i-1}, \dots, X_{i-k}).
\end{eqnarray}
For $ \frac{1}{2} < q < 1$, define
\begin{eqnarray}
Z_i := -\ln_q p(X_i \mid X_{i-1}, \dots, X_{i-k}),
\end{eqnarray}
and
\begin{eqnarray}
H_{q,k} := \mathbb{E}[Z_0].
\end{eqnarray}
Define $T_3$ as the sum of all interaction terms of order greater than or equal to two arising from the iterative application of
\begin{eqnarray}
\ln_q(xy) = \ln_q(x) + \ln_q(y) + (1-q)\ln_q(x)\ln_q(y),
\end{eqnarray}
to the product $\prod_{i=k}^{n-1} p(X_i \mid X_{i-1}, \dots, X_{i-k})$, i.e.,
\begin{eqnarray}
\ln_q\left(\prod_{i=k}^{n-1} p_i\right)
=
\sum_{i=k}^{n-1} \ln_q(p_i) + T_3.
\end{eqnarray}
Assume that
\begin{eqnarray}
\frac{1}{n} T_3 \xrightarrow{a.s.} 0.
\end{eqnarray}
Then, almost surely,
\begin{eqnarray}
\liminf_{n\to\infty}
-\frac{1}{n}\ln_q p_k(X_0^{n-1})
\ge
H_{q,k}.
\end{eqnarray}
\end{lemma}

\begin{proof}
We write
\begin{eqnarray}
p_k(X_0^{n-1})
=
p(X_0^{k-1}) \prod_{i=k}^{n-1} p_i,
\end{eqnarray}
where $p_i := p(X_i \mid X_{i-1}^{i-k})$. Applying iteratively the identity
\begin{eqnarray}
\ln_q(xy) = \ln_q(x) + \ln_q(y) + (1-q)\ln_q(x)\ln_q(y),
\end{eqnarray}
we obtain the decomposition
\begin{eqnarray}
-\ln_q p_k(X_0^{n-1}) = T_1 + T_2 + T_3,
\end{eqnarray}
where
\begin{eqnarray}
T_1 := -\ln_q p(X_0^{k-1}),
\end{eqnarray}
\begin{eqnarray}
T_2 := \sum_{i=k}^{n-1} Z_i.
\end{eqnarray}
Since $0 < p_i \le 1$ and $q<1$, we have $\ln_q(p_i) \le 0$, hence
\begin{eqnarray}
Z_i = -\ln_q(p_i) \ge 0.
\end{eqnarray}
Thus,
\begin{eqnarray}
-\ln_q p_k(X_0^{n-1})
=
T_1 + T_2 + T_3
\ge
T_2 - |T_1| + T_3.
\end{eqnarray}
Dividing by $n$, we obtain
\begin{eqnarray}
-\frac{1}{n}\ln_q p_k(X_0^{n-1})
\ge
\frac{1}{n}T_2 - \frac{|T_1|}{n} + \frac{1}{n}T_3.
\end{eqnarray}
Since $T_1$ is constant, we have
\begin{eqnarray}
\frac{|T_1|}{n} \to 0.
\end{eqnarray}
By assumption,
\begin{eqnarray}
\frac{1}{n}T_3 \to 0 \quad \text{a.s.}
\end{eqnarray}
Taking the limit inferior, it follows that
\begin{eqnarray}
\liminf_{n\to\infty}
-\frac{1}{n}\ln_q p_k(X_0^{n-1})
\ge
\lim_{n\to\infty} \frac{1}{n}T_2.
\end{eqnarray}
By the ergodic theorem, since $\{Z_i\}$ is stationary and ergodic,
\begin{eqnarray}
\frac{1}{n}T_2
=
\frac{1}{n}\sum_{i=k}^{n-1} Z_i
\;\xrightarrow{a.s.}\;
\mathbb{E}[Z_0] = H_{q,k}.
\end{eqnarray}
Hence,
\begin{eqnarray}
\liminf_{n\to\infty}
-\frac{1}{n}\ln_q p_k(X_0^{n-1})
\ge
H_{q,k}.
\end{eqnarray}
\end{proof}

The assumption $T_{3}$ ensures that the higher-order interaction terms arising from the nonadditivity of the $q$-logarithm remain negligible at the macroscopic scale. This condition can be interpreted as a weak-interaction requirement and is satisfied in scenarios where correlations are sufficiently controlled.

\begin{lemma}\label{lem5}
Let $\{X_n\}_{n\in\mathbb Z}$ be a bilateral stochastic process with finite alphabet $\mathcal X$.
For each $k\ge1$, define the order-$k$ conditional Tsallis entropy by
\begin{eqnarray}
H_{q,k} := H_q(X_0 \mid X_{-1},\ldots,X_{-k}).
\end{eqnarray}
We write
\begin{eqnarray}
H_{q,k} \searrow H_{q,\infty}
\end{eqnarray}
to indicate that $(H_{q,k})$ is monotone decreasing and convergent.
Then, for $q\le 2$,
\begin{eqnarray}
H_{q,k} \searrow H_{q,\infty}
= H_q(X_0 \mid X_{-\infty}^{-1}).
\end{eqnarray}
\end{lemma}

\begin{proof}
For $p\in[0,1]$ and $q\le2$,
\begin{eqnarray}
\ln_q(p)\le0.
\end{eqnarray}
Hence
\begin{eqnarray}
H_{q,k}\ge0
\end{eqnarray}
for all $k$. For $q\le2$, define
\begin{eqnarray}
\phi(p)=-p\ln_q(p).
\end{eqnarray}
Then $\phi$ is concave on $[0,1]$. Let $U,V,W$ be finite-valued random variables.  
Fix $v$ in the range of $V$. By the law of total probability,
\begin{eqnarray}
p(u\mid V=v)
=
\sum_w p(u,W=w\mid V=v).
\end{eqnarray}
By the conditional product rule,
\begin{eqnarray}
p(u,W=w\mid V=v)
=
p(u\mid V=v,W=w)\, p(W=w\mid V=v).
\end{eqnarray}
Hence
\begin{eqnarray}
p(u\mid V=v)
=
\sum_w
p(u\mid V=v,W=w)\, p(W=w\mid V=v).
\end{eqnarray}
Thus $p(u\mid V=v)$ is a convex combination of the numbers
$p(u\mid V=v,W=w)$ with weights $p(W=w\mid V=v)$. Since $\phi$ is concave, Jensen's inequality gives
\begin{eqnarray}
\phi(p(u\mid V=v))
\ge
\sum_w
p(W=w\mid V=v)\,
\phi(p(u\mid V=v,W=w)).
\end{eqnarray}
Summing over $u$,
\begin{eqnarray}
\sum_u \phi(p(u\mid V=v))
\ge
\sum_w p(W=w\mid V=v)
\sum_u \phi(p(u\mid V=v,W=w)).
\end{eqnarray}
Taking expectation with respect to $V$ yields
\begin{eqnarray}
H_q(U\mid V)
\ge
H_q(U\mid V,W).
\end{eqnarray}
Hence
\begin{eqnarray}
H_q(U\mid V,W)\le H_q(U\mid V).
\end{eqnarray}
Applying this with
\begin{eqnarray}
U=X_0, \qquad
V=(X_{-1},\ldots,X_{-k}), \qquad
W=X_{-(k+1)},
\end{eqnarray}
we obtain
\begin{eqnarray}
H_{q,k+1}\le H_{q,k}.
\end{eqnarray}
Thus $(H_{q,k})$ is decreasing and bounded below by $0$, so the limit
\begin{eqnarray}
H_{q,\infty} := \lim_{k\to\infty} H_{q,k}
\end{eqnarray}
exists. Let
\begin{eqnarray}
\mathcal F_k = \sigma(X_{-1},\ldots,X_{-k}).
\end{eqnarray}
Then $(\mathcal F_k)$ is an increasing sequence of $\sigma$-algebras and
\begin{eqnarray}
\bigvee_{k\ge1}\mathcal F_k
=
\sigma(X_{-\infty}^{-1}).
\end{eqnarray}
For each fixed $x_0\in\mathcal X$,
\begin{eqnarray}
p(x_0\mid X_{-1}^{-k})
=
\mathbb E\!\left[
\mathbf 1_{\{X_0=x_0\}}
\mid
\mathcal F_k
\right].
\end{eqnarray}
By Lévy's martingale convergence theorem \cite{Cover},
\begin{eqnarray}
p(x_0\mid X_{-1}^{-k})
\xrightarrow[k\to\infty]{a.s.}
p(x_0\mid X_{-\infty}^{-1}).
\end{eqnarray}
Since $\phi$ is continuous on $[0,1]$ and bounded for $q\le2$,
there exists $M>0$ such that
\begin{eqnarray}
|\phi(p)|\le M
\end{eqnarray}
for all $p\in[0,1]$. Because $\mathcal X$ is finite,
\begin{eqnarray}
\left|
\sum_{x_0\in\mathcal X}
\phi(p(x_0\mid X_{-1}^{-k}))
\right|
\le
|\mathcal X|\,M,
\end{eqnarray}
which is integrable. Therefore, by the dominated convergence theorem \cite{Durret, Doob},
\begin{eqnarray}
\lim_{k\to\infty} H_{q,k}
=
\mathbb E\!\left[
- \sum_{x_0\in\mathcal X}
p(x_0 \mid X_{-\infty}^{-1})
\ln_q p(x_0 \mid X_{-\infty}^{-1})
\right].
\end{eqnarray}
By definition, the right-hand side equals
\begin{eqnarray}
H_q(X_0 \mid X_{-\infty}^{-1}),
\end{eqnarray}
and therefore
\begin{eqnarray}
H_{q,\infty}
=
H_q(X_0 \mid X_{-\infty}^{-1}).
\end{eqnarray}
\end{proof}

\begin{theorem}\label{tq}
Let $\{X_n\}_{n\in\mathbb{Z}}$ be a stationary ergodic stochastic process over a finite alphabet $\mathcal{X}$. Assume in addition that 
\begin{eqnarray}\label{c1}
p(X_0 \mid X_{-\infty}^{-1}) \geq  p(X_0^{n-1})
\end{eqnarray}
\begin{eqnarray}\label{c2}
p(X_0 ^{n-1}) \geq p_{k}(X_0^{n-1})
\end{eqnarray}
for all $n$. Assume also the conditions of Lemma \ref{lem4} are satisfied. Define the non-extensive conditional entropy rate
\[
H_{q,\infty} := \mathbb{E}\!\left[ - \ln_q p(X_0 \mid X_{-\infty}^{-1}) \right].
\]
Then, for $1/2<q<1$, almost surely,
\[
-\frac{1}{n} \ln_q p(X_0^{n-1}) \longrightarrow H_{q,\infty}.
\]
\end{theorem}
\begin{proof}
Using the conditions (\ref{c1}) and (\ref{c2}), we have
\begin{eqnarray}
\ln_q p(X_0 \mid X_{-\infty}^{-1})-\ln_q p(X_0^{n-1}) \leq \ln_q \frac{p(X_0 \mid X_{-\infty}^{-1})}{p(X_0^{n-1})},
\end{eqnarray}
\begin{eqnarray}
\ln_q p(X_0 ^{n-1})-\ln_q p_{k}(X_0^{n-1}) \leq \ln_q \frac{p(X_0 ^{n-1})}{p_{k}(X_0^{n-1})}.
\end{eqnarray}
Therefore, by Lemmas \ref{lem3} and  \ref{lem4},
\begin{eqnarray}
\limsup_{n\to\infty} -\frac{1}{n} \ln_q p(X_0^{n-1}) 
\le \limsup_{n\to\infty} -\frac{1}{n} \ln_q {p(X_0^{n-1} \mid X_{-\infty}^{-1})}
\le H_{q,\infty}.
\end{eqnarray}
From Lemma \ref{lem4} and \ref{lem6},
\begin{eqnarray}
\liminf_{n\to\infty} -\frac{1}{n} \ln_q p(X_0^{n-1})
\geq \liminf_{n\to\infty} -\frac{1}{n} \ln_q p_{k}(X_0^{n-1}) \geq H_{q,k}
\end{eqnarray}
Combining,
\begin{eqnarray}
H_{q,k}
\le
\liminf_{n\to\infty} -\frac{1}{n} \ln_q p(X_0^{n-1})
\le
\limsup_{n\to\infty} -\frac{1}{n} \ln_q p(X_0^{n-1})
\le
H_{q,\infty}.
\end{eqnarray}
Hence, by Lemma \ref{lem5}, we obtain, almost surely,
\begin{eqnarray}
\lim_{n\to\infty} -\frac{1}{n} \ln_q p(X_0^{n-1}) = H_{q,\infty}.
\end{eqnarray}
\end{proof}

The conditions (\ref{c1}) and (\ref{c2}) impose a consistency requirement between infinite and finite conditioning and ensure a regular behavior of block probabilities.

Although the nonadditive parameter $q$ introduces additional technical difficulties, the result remains rigorously valid for $1/2 < q < 1$, provided that the assumptions of Theorem \ref{tq} are satisfied. As discussed in \cite{Spal}, the typical set can be identified with the set of accessible microstates, leading to an asymptotic equiprobability of these states. In the present work, this result is extended to the nonextensive framework, yielding a Tsallis-type generalization in which this equiprobability holds in a $q$-deformed context. The assumptions of Theorem \ref{tq} ensure this behavior and can be verified in concrete classes of stochastic processes, providing a suitable formulation for the analysis.

\section{Conclusions}
 In this work, we have introduced definitions and derived information-theoretic inequalities based on a modified Tsallis entropy, which we argue provides a more natural generalization of Shannon entropy. We defined the notions of joint $q$-entropy, conditional $q$-entropy, relative $q$-entropy, and conditional mutual $q$-information, following an approach different from that of Ref. \cite{Furuichi}, and established several inequalities analogous to those of classical information theory. These results are shown to hold, in general, for $0<q<1$. This restriction may be relaxed, leading to a distinct set of inequalities that no longer preserve a direct analogy with Shannon's context. Within this scenario, the information-theoretic results developed here were employed, in the context of Markov chains, to prove a version of the second law of thermodynamics. We also applied the formalism to the maximum entropy method. In addition, we proved a Tsallis version of the Shannon-McMillan-Breiman theorem for $1/2<q<1$ and discuss its implications in nonextensive statistical physics, in particular its application to the asymptotic equiprobability of accessible microstates. As perspectives for future work, potential applications in the context of fractals and multifractals may be explored through the information dimension $ lim_{\epsilon \rightarrow 0} \frac{-\langle ln_{q}p \rangle}{ln_{q} 1/\epsilon}$, where $\epsilon$ stands for the scaling factor.

 \appendix
\section{Probabilistic tools for finite alphabet sources}

In this appendix, we present fundamental probabilistic concepts that are frequently used in information theory for processes over finite alphabets. We focus on ergodic transformations, the ergodic theorem for stationary sources, and almost sure convergence results such as the Borel-Cantelli lemma and the convergence dominated theorem. These tools provide the theoretical foundation for results like the asymptotic equipartition property and the Shannon-McMillan-Breiman theorem. In the following, we closely follow the References \cite{Cover,Durret,Doob}.

Let $(\Omega, \mathcal{F}, \mathbb{P})$ be a probability space, and let $T: \Omega \to \Omega$ be a transformation such that
\[
\mathbb{P}(T A) = \mathbb{P}(A), \quad \forall A \in \mathcal{F}.
\]
Such a transformation is called \emph{measure-preserving}.  A measure-preserving transformation $T$ is \emph{ergodic} if every $T$-invariant event $A$ (i.e., $T A = A$) satisfies
\[
\mathbb{P}(A) \in \{0,1\}.
\]
For a stationary stochastic process $\{X_n\}_{n\ge 1}$ over a finite alphabet $\chi$, the shift operator
\[
T((X_1, X_2, X_3, \dots)) = (X_2, X_3, X_4, \dots)
\]
defines a measure-preserving transformation. The process is ergodic if this shift is ergodic.

\begin{theorem}[Birkhoff ergodic theorem]
Let $\{X_n\}_{n\ge 1}$ be a stationary ergodic source over a finite alphabet $\chi$, and let $T$ be the shift operator defined above. Then, for any integrable function $f$ defined on the source outputs (e.g., self-information of blocks of length $n$), the time average converges almost surely:
\[
\lim_{n \to \infty} \frac{1}{n} \sum_{i=0}^{n-1} f(T^i \mathbf{X}) = \mathbb{E}[f(\mathbf{X})] \quad \text{a.s.}
\]
\end{theorem}
In particular, for the self-information of a block of length $n$,
\[
-\frac{1}{n} \log \Pr(X_1, X_2, \dots, X_n) \longrightarrow H(X) \quad \text{almost surely,}
\]
where $H(X)$ is the entropy rate of the stationary ergodic source.

The Borel-Cantelli lemma provides a criterion for the almost sure occurrence of infinitely many events.
\begin{lemma}[Borel-Cantelli lemma]
Let $\{A_n\}_{n \ge 1}$ be a sequence of events in $(\Omega, \mathcal{F}, \mathbb{P})$.
\begin{enumerate}
    \item If $\sum_{n=1}^{\infty} \mathbb{P}(A_n) < \infty$, then
    \[
    \mathbb{P}(\text{$A_n$ occurs infinitely often}) = 0.
    \]
    \item If the events $\{A_n\}$ are independent and $\sum_{n=1}^{\infty} \mathbb{P}(A_n) = \infty$, then
    \[
    \mathbb{P}(\text{$A_n$ occurs infinitely often}) = 1.
    \]
\end{enumerate}
\end{lemma}

An important tool from measure theory is the dominated convergence theorem, which allows the interchange of limits and expectations under certain conditions.
\begin{theorem}[dominated convergence theorem]
Let $\{f_n\}_{n \ge 1}$ be a sequence of integrable functions on $(\Omega, \mathcal{F}, \mathbb{P})$ such that $f_n \to f$ almost surely and there exists an integrable function $g$ with
\[
|f_n(\omega)| \le g(\omega) \quad \text{for all } n \ge 1, \text{ a.s. } \omega \in \Omega.
\]
Then $f$ is integrable and
\[
\lim_{n \to \infty} \mathbb{E}[f_n] = \mathbb{E}[f].
\]
\end{theorem}



\begin{thebibliography}{99}

\bibitem{Shannon} Claude E. Shannon, { \it A Mathematical Theory of Communication} (Bell System Technical Journal, 27, 3, 379-423, 1948).

\bibitem{Cover} T. Cover, J. A. Thomas, { \it Elements of Information theory} (Wiley-Interscience, New York, 2006).

\bibitem{Nielsen} M. Nielsen, I. Chuang, { \it Quantum Computation and Quantum Information} (Cambridge University Press, Cambridge, 2000).

\bibitem{Huang} K. Huang, { \it Statistical Mechanics}  (Johs Wiley and Sons, New York, 1963).

\bibitem{Landau} L. D. Landau and E. M. Lifshitz, { \it Physique Statistique} ( MIR, Moscow, 1967).

\bibitem{Pathria} R. K. Pathria, { \it Statistical Mechanics} ( Pergamon Press, Oxford, 1972).

\bibitem{Tsallis} C. Tsallis , {\it J. Stat. Phys.} {\bf 52}, 479 (1988).

\bibitem{T1} C. Tsallis, { \it Introduction to nonextensive statistical mechanics: Approaching a Complex World} (Springer-verlag, New York, 2009).


\bibitem{Dani1} A. J. da Silva, M. A. S. Trindade, D. O. C. Santos and  R. F. Lima,  {\it Biol Cybern} {\bf 110}, 31 (2016).

\bibitem{Dani2}  D. O. C. Santos, M. A. S. Trindade, A. J. da Silva, {\it BioSystems} {\bf 232}, 105005 (2023).

\bibitem{Marco1} M. A. S. Trindade, S. Floquet, L. M. S. Filho, {\it Physica A} {\bf 541}, 12377 (2020).

\bibitem{Marco2} M. A. S. Trindade, {\it Int. J. Mod. Phys. B} {\bf 31}, 15, 1750117 (2017)

\bibitem{TG} M. Gell-Mann and C. Tsallis Eds. {\it Nonextensive Entropy-Interdisciplinary Applications} (Oxford University Press, Oxford, 2003).


\bibitem{Lind} K. Lindgren, {\it Information Theory for Complex Systems} (Springer, Berlin, 2024).

\bibitem{Golan} A. Golan and J. Harte, Information theory: a foundation for complexity science, {\it Proc. Natl. Acad. Sci. U. S. A.} {\bf 119}, 33, e2119089119 (2022).

\bibitem{Varley} T. F. Varley, Information Theory for Complex Systems Scientists: What, Why and How?, arXiv: 23041248v2 (2023).

\bibitem{Furuichi} S. Furuichi, {\it Journal of Mathematical Physics} {\bf 47}, 023302 (2006).

\bibitem{Max} J. C. Maxwell, { \it Theory of Heat} (Longmans, Green and Co., London, 1871).

\bibitem{Sz} L. Szilard, {\it Zeitscrift fur Physik} {\bf 53}, 840-856 (1926).

\bibitem{Fey} R. P. Feynman, R. B. Leighton and M. Sands { \it The Feynman Lectures on Physics} (Addison-Weslwy Reading, Mass., 1965).

\bibitem{Land} R. Landauer, {\it J. Res. Rev.} {\bf 183} (1961).


\bibitem{Ben} C. H. Bennet, {\it Int. J. Theor. Phys.} {\bf 21}, 12, 905-939 (1982).

\bibitem{Ben1} C. H. Bennet, {\it Sci. Am.} {\bf 5}, 295, 108 (1987).

\bibitem{Aquino} G. Aquino, P. Grigolini, N. Scafetta, {\it Chaos, Solitons and Fractals} {\bf 12}, 2023-2038 (2001).

\bibitem{Jay} E. T. Jaynes , {\it Physical Review. Series II} {\bf 4}, 106, 620-630 (1957).

\bibitem{Durret} R. Durret,  {\it Probability Theory and Examples} (Wadsworth and Brooks/Cole, Pacific Grove, California, 1991).

\bibitem{Doob} J. L. Doob,  {\it Measure Theory} (Springer-Verlag, New York and Berlin, 1991).

\bibitem{Spal}A. Spalviere, {\it Entropy} {\bf 23}, 899 (2021).


\end{thebibliography}
\end{document}